\documentclass[10pt]{article}
\usepackage{amsmath,amsfonts}
\normalbaselines
\newtheorem{teor}{Theorem}

\newtheorem{lema}[teor]{Lemma}
\newtheorem{coro}[teor]{Corollary}
\newtheorem{rem}[teor]{Remark}

\newcommand{\caja}{\null\hfill\rule{2mm}{2mm}}

\title{Scalar curvature of spacelike hypersurfaces and certain class of cosmological models for accelerated expanding universes}

\author{Juan A. Aledo${}^{a}$ and Rafael M. Rubio${}^{b}$ \\[6mm]
${}^a$ Departamento de Matem\'aticas, E.S.I. Inform\'atica, \\[0.5mm] Universidad de
Castilla-La Mancha, 02071 Albacete, Spain,\\ E-mail\textup{:
\texttt{juanangel.aledo@uclm.es}} \\[3mm]
${}^b$ Departamento de Matem\'aticas, Campus de Rabanales, \\[0.5mm] Universidad de
C\'ordoba, 14071 C\'ordoba, Spain,\\[0.5mm] E-mail\textup{: \texttt{rmrubio@uco.es}}\\[3mm]}

\date{}

\begin{document}

\maketitle

\thispagestyle{empty}

\begin{abstract}
We study the scalar curvature of spacelike hypersurfaces in the
family of cosmological models known as generalized Robertson-Walker
spacetimes, and give several rigidity results under appropriate
mathematical and physical assumptions. On the other hand, we show
that this family of spacetimes provides suitable models obeying the
null convergence condition to explain accelerated expanding universes.
\end{abstract}

\vspace*{5mm}



\section{Introduction}
In this paper we deal with the class of cosmological models called
\emph{generalized Robertson-Walker (GRW) spacetimes} (see Section
\ref{s2}), which are warped products $I \times_f F$ with base an
open interval $(I,-dt^2)$ and  fiber a Riemannian manifold
$(F,g_{_F})$ whose sectional curvature is not assumed to be
constant. Thus, our ambient spacetimes widely extend to those that
are classically called Robertson-Walker (RW) spacetimes. Recall that
the class of Robertson-Walker spacetimes includes the usual big-bang
cosmological models, the de Sitter spacetime, the steady state
spacetime, the Lorentz-Minkowsky spacetime and the Einstein's static
spacetime, among others. Unlikely to these spacetimes, our ambient
spacetimes are not necessarily spatially-homogeneous. Note that
being spatially-homogeneous, which is reasonable as a first
approximation of the large scale structure of the universe, could
not be appropriate when we consider a more accurate scale. Thus, a
GRW spacetime could be a suitable spacetime to model a universe with
inhomogeneous spacelike geometry \cite{Ra-Sch}. On the other hand,
small deformations of the metric on the fiber of classical
Robertson-Walker spacetimes fit into the class of GRW spacetimes.
Therefore, GRW spacetimes are useful to analyze if a property of a
RW spacetime $\overline{M}$ is \emph{stable}, i.e.  if it remains
true for spacetimes close to $\overline{M}$ in a certain topology
defined on a suitable family of spacetimes \cite{Geroch}. In fact, a
deformation $s \mapsto g_{_F}^{(s)}$ of the metric of $F$ provides a
one parameter family of GRW spacetimes close to $\overline{M}$ when
$s$ approaches to $0$. Note that a conformal change of the metric of
a GRW spacetime with a conformal factor which only depends on $t$,
produces a new GRW spacetime.  Any GRW spacetime has a
smooth global time function, and so it is \emph{stably causal}
\cite[p. 64]{Beem}. Moreover, if the fiber is complete then the GRW
spacetime is \emph{globally hyperbolic} \cite[Th. 3.66]{Beem}. On
the other hand, if the fiber is compact then it is called
\emph{spatially closed}. In \cite{Sanchez} the behaviour of the
geodesics of GRW spacetimes is studied.

We will impose the spacetime to obey the \emph{null convergence
condition (NCC)}, which says that the  Ricci tensor of the spacetime
is semi-definite positive on every null (light-like)  vector. Recall
that the exact solutions to the Einstein equations with cosmological
constant, provided that the stress-energy momentum tensor satisfies
the weak energy condition, obey the null convergence condition.

On the other hand, the study of spacelike hypersurfaces in General
Relativity is relevant for several questions,  as foliations of
spacetimes, change of expansion or contraction phases, the Cauchy
problem for Einstein's equation, etc. (see, for instance,
\cite{M-T}, \cite{Cho-l}). Moreover, for many problems in General
Relativity, including the \emph{Positive Mass Theorem} and the
\emph{Penrose Inequality}, knowledge of the entire spacetime is not
necessary, rather attention may be focused solely on a spacelike
hypersurface, playing the scalar curvature of this hypersurface an
important role \cite{Sch}. In addition, the choice of a constant
mean curvature (CMC) spacelike hypersurface as initial data have
been considered in order to deal with the Cauchy problem for the
Einstein's equation (see \cite{Cho}).

In the first part of this paper (Section \ref{s3}) we study the
scalar curvature of spacelike hypersurfaces in a GRW spacetime which
obeys the NCC (see Lemma \ref{l1}). Thus, we obtain a general
expression for the scalar curvature of an immersed spacelike
hypersurface in such an ambient space (\ref{escalar}), given several
estimations when the spacetime obeys the NCC and characterizing
those spacelike hypersurfaces which attain the equality in our
estimations (Theorem \ref{TF} and Corollary \ref{conS}). In this
setting, we pay a special attention to the important case of maximal
hypersurfaces (Corollaries \ref{cmax} and \ref{cmax2}). As a
consequence of our results, in the particular case when the
spacetime is the de Sitter space we provide a characterization of
the totally umbilical spacelike hypersurfaces from a bound of the
scalar curvature of the hypersurface (see Theorem \ref{TDS} and
Remark \ref{rds}). We also particularize our study to the case of
compact CMC hypersurfaces in a GRW spacetime which obeys the NCC, so
obtaining more strong consequences including a Calabi-Bernstein type
result (Theorem \ref{tcom}).

On the other hand, in the second part of the paper (Section
\ref{s4}) we apply our mathematical results to the study of a
certain class of cosmological models, specifically GRW spacetimes
filled with perfect fluid. In General Relativity one often employs a
perfect fluid stress-energy momentum tensor to represent the source
of the gravitational field. This fluid description is used where one
assumes that the large-scale proprieties of the universe can be
studied by assuming a perfect fluid description of the sources. A
review of the specifical literature shows that, in fact, almost all
the cosmological studies use the perfect fluid model. We focus on
the case where GRW spacetimes satisfying the NCC constitute perfect
fluid models adequate to describe universes at dominant dark energy
stage, namely, accelerated expanding universes. We end up
particularizing our study to the family of spatially closed GRW
spacetimes. In this setting, we are able to express the total energy
on a compact spacelike hypersurface in terms of its scalar and mean
curvatures (Theorems \ref{tf1} and \ref{tf2}). Finally, in the
simplest case of a 3-dimensional GRW spacetime, as a  consequence of
the Gauss-Bonnet theorem we provide a nice expression of the total
energy in terms of the Euler characteristic of the surface, its mean
curvature and its volume (Theorem \ref{GB}).

\section{Preliminaries}
\label{s2} \noindent

Let $(F,g_{_F})$ be an $n(\geq 2)$-dimensional (connected)
Riemannian manifold, $I$ an open interval in $\mathbb{R}$ endowed
with the metric $-dt^2$, and  $f$ a positive smooth function defined
on $I$. Then, the product manifold $I \times F$ endowed with the
Lorentzian metric
\begin{equation}\label{metrica}
\bar{g} = -\pi^*_{_I} (dt^2) +f(\pi_{_I})^2 \, \pi_{_F}^* (g_{_F})
\, ,
\end{equation}
where $\pi_{_I}$ and $\pi_{_F}$ denote the projections onto $I$ and
$F$, respectively, is called a \emph{Generalized Robertson-Walker
(GRW) spacetime} with \emph{fiber} $(F,g_{_F})$, \emph{base}
$(I,-dt^2)$ and \emph{warping function} $f$. Along this paper we
will represent this $(n+1)$-dimensional Lorentzian manifold by
$\overline{M}= I \times_f F$.

The coordinate vector field $\partial_t:=\partial/\partial t$
globally defined on $\overline{M}$ is (unitary) timelike, and so
$\overline{M}$ is time-orientable. We will also consider on
$\overline{M}$ the conformal closed timelike vector field $K: =
f({\pi}_I)\,\partial_t$. From the relationship between the
Levi-Civita connections of $\overline{M}$ and those of the base and
the fiber \cite[Cor. 7.35]{O'N}, it follows that
\[
\overline{\nabla}_XK = f'({\pi}_I)\,X
\]
for any $X\in \mathfrak{X}(\overline{M})$, where $\overline{\nabla}$
is the Levi-Civita connection of the Lorentzian metric
(\ref{metrica}).

We will denote by $\overline{{\rm Ric}}$ and $\overline{S}$ the
Ricci tensor and the scalar curvature of $\overline{M}$,
respectively. It is a straightforward computation (see \cite[Cor.
7.43]{O'N}) to check that
\begin{equation}\label{Ricb}
\overline{{\rm Ric}}(X,Y) ={\rm
Ric}^F(X^F,Y^F)+\left(\frac{f''}{f}+(n-1)\frac{f'^2}{f^2}\right)
\overline{g}(X^F,Y^F)-n\, \frac{f''}{f} \overline{g}(X,\partial_t)
\overline{g}(Y,\partial_t)
\end{equation}
for $X,Y\in \mathfrak{X}(\overline{M})$, where ${\rm Ric}^F$ stands
for the Ricci tensor of $F$. Here $X^F$ denotes the lift of the
projection of the vector field $X$ onto $F$, that is,
\begin{equation}\label{e1}
X =X^F-\overline{g}(X,\partial_t)\partial_t.
\end{equation}

Recall  that a Lorentzian manifold $\overline{M}$ obeys the
\emph{Null Convergence Condition (NCC)} if its Ricci tensor
$\overline{\rm Ric}$ satisfies $\overline{\rm Ric}(X,X) \geq 0$, for
all null vector $X\in \mathfrak{X}(\overline{M})$. If $\overline{M}=
I \times_f F$, given a null vector field $X\in
\mathfrak{X}(\overline{M})$ we get, by decomposing $X$ as in
(\ref{e1}), that
$\overline{g}(X^F,X^F)=\overline{g}(X,\partial_t)^2\neq 0$. Then, if
we denote by $\overline{Ric}$ and $Ric^F$ the Ricci curvatures of
$\overline{M}$ and $(F,g_{_F})$ respectively, we get from
(\ref{Ricb})
\[
\overline{Ric}(X) =\left|X^F\right|_{_F}^2 \left(
Ric^F\left(X^F\right)-(n-1)f^2 (\log f)''\right)
\]
for all unitary null vector field $X\in \mathfrak{X}(\overline{M})$,
where $\left|X^F\right|_{_F}=g_{_F}(X^F,X^F)^{1/2}$ and 

$$Ric^F(X^F)=\frac{{\rm Ric}^F(X^F,X^F)}{g_{_F}(X^F,X^F)}={\rm
Ric}^F\Big(\frac{X^F}{\mid X^F\mid_{_F}},\frac{X^F}{\mid
X^F\mid_{_F}}\Big).$$

Therefore, we have:
\begin{lema} \label{lNCC} Let   $\overline{M}= I \times_f F$ be a GRW spacetime. Then,
$\overline{M}$ obeys the NCC if and only if
\begin{equation}\label{NCCc}
Ric^F-(n-1)f^2 (\log f)''\geq 0.
\end{equation}
\end{lema}

\vspace{3mm}

Observe that as a consequence of the previous Lemma  $\overline{M}$ obeys the NCC if and only at each point $p\in \overline{M}$ the Ricci curvature at any direction of $T_p\overline{M}$ is greater or equal  than $(n-1)f^2 (\log f)''\geq 0$, for all $p\in (n-1)f^2 (\log f)''\geq 0$.

Regarding the scalar curvature $\overline{S}$ of $\overline{M}$, we
get from (\ref{Ricb}) that
\begin{equation}\label{Sb}
\overline{S}={\rm trace}\left(\overline{{\rm
Ric}}\right)=\frac{S^F}{f^2}+2n\frac{f''}{f}+n(n-1)\frac{f'^2}{f^2}
\end{equation}
where $S^F$ stands for the scalar curvature of $F$.

Given an $n$-dimensional manifold $M$, an immersion $\psi: M
\rightarrow \overline{M}$ is said to be \emph{spacelike} if the
Lorentzian metric (\ref{metrica}) induces, via $\psi$,  a Riemannian
metric $g$ on $M$. In this case, $M$ is called a \emph{spacelike
hypersurface}.

Since $\overline{M}$ is time-orientable we can take, for each
spacelike hypersurface $M$ in $\overline{M}$, a unique unitary
timelike vector field $N \in \mathfrak{X}^\bot(M)$ globally defined
on $M$ with the same time-orientation as $\partial_t$, i.e. such
that $\bar{g}(N,\partial_t)<0$. From the wrong-way Cauchy-Schwarz
inequality (see \cite[Prop. 5.30]{O'N}, for instance), we have
$\bar{g}( N,
\partial_t) \leq -1$, and the equality holds at a point $p\in M$ if
and only if $N = \partial_t$ at $p$.

For a spacelike hypersurface $\psi: M \rightarrow \overline{M}$ with
Gauss map $N$, the \emph{hyperbolic angle} $\varphi$, at any point
of $M$, between the unit timelike vectors $N$ and $\partial_t$, is
given by $\bar{g}(N,\partial_t)=-\cosh \varphi$. By simplicity,
throughout this paper we will refer to $\varphi$  as the
\emph{hyperbolic angle function} on $M$. In a GRW spacetime
$\overline{M}$ the integral curves of $\partial_t$ are called
\emph{comoving observers} \cite[p. 18]{SW}. If $p$ is a point of a
spacelike hypersurface $M$ in $\overline{M}$, among the
instantaneous observers at $p$, $\partial_t(p)$ and $N_{_p}$ appear
naturally. In this sense, observe that the energy $e(p)$ and the
speed $v(p)$ that $\partial_t(p)$ measures for $N_{_p}$ are given,
respectively, by $e(p)=\cosh\theta(p)$ and
$|v(p)|^2=\tanh^2\theta(p)$ \cite[pp. 45, 67]{SW}.

We will denote by $A$ and $H:= -(1/n) \mathrm{trace}(A)$ the
\emph{shape operator} and the \emph{mean curvature function}
associated to $N$. The mean curvature is zero if and only if the
spacelike hypersurface is, locally, a critical point of the
$n$-dimensional area functional for compactly supported normal
variations. A spacelike hypersurface with $H=0$ is called a
\emph{maximal} hypersurface.

In any GRW spacetime $\overline{M}$ there is a remarkable family of
spacelike hypersurfaces, namely its spacelike \emph{slices}
$\{t_{_0}\}\times F$, $t_{_0}\in I$. The spacelike slices constitute
for each value $t_0$ the restspace of the distinguished observers in
$\partial_t$. It can be easily seen that a spacelike hypersurface in
$\overline{M}$ is a (piece of) spacelike slice if and only if the
function $\tau:=\pi_I \circ \psi$ is constant. Furthermore, a
spacelike hypersurface in $\overline{M}$ is a (piece of) spacelike
slice if and only if the hyperbolic angle $\varphi$ vanishes
identically. The shape operator of the spacelike slice $\tau=t_{_0}$
is given by $A=-f'(t_{_0})/f(t_{_0})\,I$, where $I$ denotes the
identity transformation, and so its (constant) mean curvature is $H=
 f'(t_{_0})/f(t_{_0})$. Thus, a spacelike slice is maximal if and
only if $f'(t_{_0})=0$ (and hence, totally geodesic).

\section{Spacelike Hypersurfaces in a GRW which obeys the NCC}\label{s3}
Let $\psi: M \rightarrow \overline{M}$ be a spacelike hypersurface
in the GRW spacetime $\overline{M}= I \times_f F$. The curvature
tensor $R$ of $M$ can be described in terms of the curvature tensor
$\overline{R}$ of $\overline{M}$ and the shape operator $A$
according the Gauss equation
\begin{equation}
\label{EG} R(X,Y)Z=\left( \overline{R}(X,Y)Z
\right)^T-\overline{g}(AX,Z)AY+\overline{g}(AY,Z)AX
\end{equation}
for all tangent vector fields $X,Y,Z\in\mathfrak{X}(M)$, where
$\left( \overline{R}(X,Y)Z  \right)^T$ denotes the tangential
component of
\[
\overline{R}(X,Y)Z=\overline{\nabla}_{[X,Y]}Z-[\overline{\nabla}_X,\overline{\nabla}_Y]Z.
\]

From (\ref{EG}) it follows that the Ricci curvature of $M$ is given
by
\[
{\rm Ric}(X,Y)=\overline{{\rm
Ric}}(X,Y)+\overline{g}(\overline{R}(X,N)Y,N)-{\rm
trace}(A)\overline{g}(AX,Y)+\overline{g}(AX,AY)
\]
for $X,Y\in\mathfrak{X}(M)$. Then, the scalar curvature of $M$
yields
\begin{equation}
\label{escalar} S={\rm trace}({\rm
Ric})=\overline{S}+2\overline{{\rm Ric}}(N,N)+{\rm
trace}(A^2)-n^2H^2.
\end{equation}

If we put $\partial_t^T=\partial_t+\overline{g}(\partial_t,N)N$ the
tangential part of $\partial_t$ and $N^F
=N+\overline{g}(N,\partial_t)\partial_t$, it follows from
$\overline{g}(N,N)=-1=\overline{g}(\partial_t,\partial_t)$ that
\begin{equation}\label{paraE}
\left|\partial_t^T\right|^2=\overline{g}(\partial_t^T,\partial_t^T)=\overline{g}(N^F,N^F)=\left|X^F\right|^2=
\sinh^2\varphi.
\end{equation}
Then, using also (\ref{Ricb}), we get
\[
\overline{{\rm Ric}}(N,N) ={\rm Ric}^F(N^F,N^F)
-(n-1)\,\frac{f''(\tau)}{f(\tau)}\, |\partial_t^T|^2 +(n-1)\,
\frac{f'^2(\tau)}{f^2(\tau)}\, |\partial_t^T|^2 -n\,
\frac{f''(\tau)}{f(\tau)}
\]
which jointly with (\ref{Sb}) allow to rewrite the scalar curvature
(\ref{escalar}) as follows

\begin{lema}\label{l1} Let $\psi: M \rightarrow \overline{M}$ be a spacelike hypersurface
in a GRW spacetime $\overline{M}= I \times_f F$. Then the scalar
curvature of $M$ is given by
\begin{eqnarray}
S & = & \frac{S^F\circ\pi_{_F}}{f^2(\tau)}+2\left({\rm
Ric}^F(N^F,N^F)-(n-1)\, (\log f)''(\tau) \, \left|\partial_t^T\right|^2\right) \nonumber \\
&&  +n(n-1)\left(\frac{f'^2(\tau)}{f^2(\tau)}-H^2\right)
+{\rm trace}(A^2)-nH^2.\label{escalar2}
\end{eqnarray}
\end{lema}

As a consequence of Lemmas \ref{lNCC} and \ref{l1}, and using also
(\ref{paraE}), we get

\begin{teor} \label{TF} Let $\psi: M \rightarrow \overline{M}$ be a spacelike hypersurface
in a GRW spacetime $\overline{M}= I \times_f F$ obeying the NCC.
Then the scalar curvature of $M$ satisfies
\[
S \geq
\frac{S^F\circ\pi_{_F}}{f^2(\tau)}+n(n-1)\left(\frac{f'^2(\tau)}{f^2(\tau)}-H^2\right).
\]
Moreover, if the equality holds then the spacelike
hypersurface is totally umbilical.
\end{teor}

From Lemma \ref{lNCC} it follows that, when the GRW spacetime obeys
the NCC, then
\begin{equation}\label{SSSS}
S^F\circ\pi_{_F}\geq n(n-1) \, f^2(\tau)\, (\log f)''(\tau).
\end{equation}
Therefore
we get the following consequence from Theorem \ref{TF}

\begin{coro} \label{conS}
Let $\psi: M \rightarrow \overline{M}$ be a spacelike hypersurface
in a GRW spacetime $\overline{M}= I \times_f F$ obeying the NCC.
Then the scalar curvature of $M$ satisfies
\begin{equation} \label{conS2}
S \geq n(n-1)\left(\frac{f''(\tau)}{f(\tau)}-H^2\right).
\end{equation}
Moreover, if the equality holds then the spacelike hypersurface is
totally umbilical and the scalar curvature of the fiber can be
expressed in terms of the warping function $f$ as $S^F\circ\pi_{_F}=
n(n-1)\, f^2(\tau)\, (\log f)''(\tau)$.
\end{coro}

For the important particular case of maximal surfaces, Theorem
\ref{TF} and Corollary \ref{conS} can be rewritten as follows
\medskip

\begin{coro} \label{cmax} Let $\psi: M \rightarrow \overline{M}$ be a maximal spacelike hypersurface
in a GRW spacetime $\overline{M}= I \times_f F$ obeying the NCC.
Then the scalar curvature of $M$ satisfies
\begin{equation} \label{paraslice}
S \geq \frac{S^F\circ\pi_{_F}}{f^2(\tau)}.
\end{equation}
Moreover, the equality holds if and only if the spacelike
hypersurface is totally geodesic.
\end{coro}

Note that if the equality holds in (\ref{paraslice}) it must be
$f'(\tau)\equiv 0$. Hence, in the particular case when the warping
function $f$ is non-locally constant, that is, when the GRW
spacetime $\overline{M}$ is \emph{proper}, we can state

\begin{coro} \label{cmax2} Let $\psi: M \rightarrow \overline{M}$ be a maximal spacelike hypersurface
in a proper GRW spacetime $\overline{M}= I \times_f F$ obeying the
NCC. Then the scalar curvature of $M$ satisfies
\[
S \geq \frac{S^F\circ\pi_{_F}}{f^2(\tau)}.
\]
Moreover, the equality holds if and only if the spacelike
hypersurface is contained in a totally geodesic slice.
\end{coro}

Corollary \ref{conS} can be easily adapted for maximal spacelike hypersurfaces by taking $H=0$. Note that under this additional hypothesis
if the equality holds in (\ref{conS2}) then the spacelike hypersurface is totally geodesic.

\begin{rem}
When $I=\mathbb{R}$, $F=\mathbb{R}^n$ and $f(t)=e^t$, the GRW
spacetime $\mathcal{N}=\mathbb{R}\times_{e^t} \mathbb{R}^n$ is
isometric to a proper open subset of the De Sitter spacetime of
sectional curvature 1, which is called the $(n+1)$-dimensional
\emph{steady state spacetime}. The steady state is the model of the
universe proposed by Bondi and Gold \cite{Bo} and Hoyle \cite{Ho}
when one is looking for a model of the universe which looks the same
not only at all points and in all directions (that is, spatially
isotropic and homogeneous), but also at all times \cite[Section
14.8]{W}. Since the steady state spacetime is a proper GRW spacetime
whose fiber has zero scalar curvature, from Corollary \ref{cmax} we
can state that any maximal spacelike
hypersurfaces in $\mathcal{N}$ has non negative scalar curvature
$S\geq 0$, and moreover it can not be $S=0$ up to at isolated
points.
\end{rem}


In the language of General Relativity, the de Sitter spacetime is
the maximally symmetric vacuum solution of Einstein's field
equations with a positive (repulsive) cosmological constant
corresponding to a positive vacuum energy density and negative
pressure. In its intrinsic version, the $(n+1)$-dimensional De
Sitter spacetime $\mathbb{S}^{n+1}_1$ is given as the
Robertson-Walker spacetime
$\mathbb{S}^{n+1}_1=\mathbb{R}\times_{\cosh t}\mathbb{S}^ n$, where
$\mathbb{S}^ n$ denotes the $n$-dimensional sphere with its usual
metric.

Recall that $\mathbb{S}^{n+1}_1$ has constant sectional curvature
$1$ and obeys the NCC. Even more, since the fiber $F=\mathbb{S}^ n$
has constant Ricci curvature $Ric^F=n-1$, we have that equality
holds in (\ref{NCCc}). Therefore, using also that $S^F=n(n-1)$, we get from Lemma \ref{l1}
that the scalar curvature of a spacelike hypersurface in
$\mathbb{S}^{n+1}_1$ is given by
\begin{eqnarray}
S & = & \frac{n(n-1)}{\cosh^2 \tau}+n(n-1)\left(\tanh^2
\tau-H^2\right) + {\rm
trace}(A^2)-nH^2 \nonumber \\
& = & n(n-1)(1-H^2)+ {\rm trace}(A^2)-nH^2 \nonumber
\end{eqnarray}

Hence, we have
\begin{teor} \label{TDS} Let $\psi: M \rightarrow \mathbb{S}^{n+1}_1$ be a spacelike hypersurface
in the de Sitter spacetime. Then the scalar curvature of $M$
satisfies
\begin{equation}\label{desDS}
S\geq n(n-1)(1-H^2).
\end{equation}
Moreover, the equality holds if and only if the spacelike
hypersurface is totally umbilical.
\end{teor}

Recall that every totally umbilical spacelike in the De Sitter space
has constant mean curvature. In particular, if the equality holds in
(\ref{desDS}) then the spacelike hypersurface has constant scalar
curvature.

\begin{rem}\label{rds}
The problem of characterizing the totally umbilical spacelike
hypersurfaces in $\mathbb{S}^{n+1}_1$ under hypotheses relative to
the scalar curvature of the hypersurface has received an special
attention in the last years \cite{AA1}, \cite{AAR}, \cite{AR},
\cite{BCP}, \cite{CCS}, \cite{CI}, \cite{HSZ}, \cite{Li},
\cite{Liu}, \cite{Ze1}, \cite{Ze2}. It is worth pointing out that,
in these papers, to obtain the rigidity results the completeness
(and sometimes even the compactness) of the hypersurface is
required. However, in Theorem \ref{TDS} we do not need to ask the
hypersurface to be complete.
\end{rem}

We will say that a spacetime $\overline{M}$ verifies the \emph{NCC
with strict inequality} if its Ricci tensor $\overline{\rm Ric}$
satisfies $\overline{\rm Ric}(X,X) > 0$, for all null vector $X\in
\mathfrak{X}(\overline{M})$. Reasoning as in Lemma \ref{lNCC}, a GRW
spacetime $\overline{M}= I \times_f F$ obeys the NCC with strict
inequality if and only if $Ric^F-(n-1)\, f^2\, (\log
f)''> 0$.

\begin{coro} Let $\psi: M \rightarrow \overline{M}$ be a spacelike hypersurface
in a GRW spacetime $\overline{M}= I \times_f F$ obeying
the NCC with strict inequality. Then the scalar curvature of $M$ satisfies
\[
S \geq
\frac{S^F\circ\pi_{_F}}{f^2(\tau)}+n(n-1)\left(\frac{f'^2(\tau)}{f^2(\tau)}-H^2\right).
\]
Moreover, the equality holds if and only if  $M$ is contained in a
spacelike slice, and as a consequence $S=
\frac{S^F\circ\pi_{_F}}{f(\tau)^2}$.
\end{coro}
\begin{proof}
It follows from (\ref{escalar2}) and the fact that $M$ is contained
in a spacelike slice if, and only if, $\left| N^ F\right|=\left|\partial _t^T\right|=0$ on $M$.
\hfill{$\Box$}
\end{proof}

If $\overline{M}= I \times_f F$ is a spacetime obeying
NCC with strict inequality then, reasoning as in (\ref{SSSS}), we have that
\[
S^F\circ\pi_{_F}> n(n-1) \, f^2(\tau)\, (\log f)''(\tau).
\]
In particular, the scalar curvature of a spacelike hypersurface in $\overline{M}= I \times_f F$ can be bounded
in terms of the warping function $f$ as
\[
S > n(n-1)\left(\frac{f''(\tau)}{f(\tau)}-H^2\right).
\]
Unlike what was happened under the NCC, now the equality cannot be attained.

\vspace*{5mm}

It is well-known that if a GRW spacetime $\overline{M}$ admits a
compact spacelike hypersurface, then $\overline{M}$ must be
\emph{spatially closed}, that is, its fiber $F$ is compact (see
\cite[Prop.3.2]{A-R-S1}).

Let $M$ be a compact spacelike hypersurface with constant mean
curvature in $\overline{M}$. In  \cite[Sec. 4]{A-R-S1}, the authors
showed that the following integral equation holds
\begin{equation}\label{int}
\int_M \left\{\overline{{\rm Ric}}(K^{\top},
N)+\overline{g}(K,N)({\rm trace}\,(A^2)-nH^2)\right\} dV=0,
\end{equation}
where $dV$ denotes the Riemannian volume element on $M$. From
(\ref{Ricb}), and provided that $\overline{M}$ obeys the NCC, a straightforward computation yields
\begin{eqnarray*}
\overline{{\rm Ric}}(K^T,N) & = & \overline{g}(K,N) \,
\overline{{\rm Ric}}(N^F,N^F)-\overline{g}(K,N)
\left|\partial_t^T\right|^2\,
\overline{{\rm Ric}}(\partial_t,\partial_t) \nonumber \\
& = & \overline{g}(K,N)\left( {\rm Ric}^F(N^F,N^F)-(n-1)
\left|N^F\right|^2\,(\log f)''(\tau)\right) \nonumber \\
& = & \overline{g}(K,N) \left|N^F\right|^2 \left(
Ric^F\left(\frac{N^F}{|N^F|_F}\right)-(n-1) f^2(\tau)\,(\log
f)''(\tau)\right)\leq 0
\end{eqnarray*}

Since also $\overline{g}(K,N)({\rm trace}\,(A^2)-nH^2)\leq 0$, from
(\ref{int}) we get that both $\overline{{\rm Ric}}(K^T,N)$ and ${\rm
trace}\,(A^2)-nH^2$ vanish identically on $M$. Therefore $M$ is
totally umbilical and, from (\ref{escalar2}), its scalar curvature
is given by
\begin{equation}\label{pfl}
S=\frac{S^F\circ\pi_{_F}}{f^2(\tau)}+n(n-1)\left(\frac{f'^2(\tau)}{f^2(\tau)}-H^2\right).
\end{equation}

Furthermore, if $\overline{M}$ obeys NCC with strict inequality,
then $M$ is a spacelike slice. Therefore, the following
Calabi-Bernstein type result can be stated:

\begin{teor}\label{tcom}
Let $(F,g_{_F})$ be an $n$-dimensional, $n\geq 2,$ compact Riemannian manifold and
$f:I\longrightarrow (0,\infty)$ a smooth function such that
$Ric^F-(n-1)\, f^2 \, (\log f)''> 0$. Then the only
entire solutions to the mean curvature surface equation for
spacelike hypersurfaces are the constant functions.
\end{teor}

The importance in General Relativity of maximal and constant mean
curvature spacelike hypersurfaces in spacetimes is well-known; a
summary of several reasons justifying it can be found in \cite{M-T}.
In particular, hypersurfaces of non-zero constant mean curvature are
particularly suitable for studying the propagation of gravity
radiation \cite{S}. Classical papers dealing with uniqueness results for CMC hypersurfaces are
\cite{Ch}, \cite{BF} and \cite{M-T}, although a previous relevant
result in this direction was the proof of the Bernstein-Calabi
conjecture \cite{Calabi} for the $n$-dimensional Lorentz-Minkowski
spacetime given by Cheng and Yau \cite{Cheng-Yau}. In \cite{BF},
Brill and Flaherty replaced the Lorentz-Minkowski spacetime by a
spatially closed universe, and proved uniqueness in the large by
assuming $\overline{{\mathrm{Ric}}}(z,z)>0$ for all timelike vector $z$. In
\cite{M-T}, this energy condition was relaxed by Marsden and Tipler
to include, for instance, non-flat vacuum spacetimes. More recently
Bartnik proved in \cite{Bar} very general existence theorems and
consequently, he claimed that it would be useful to find new
satisfactory uniqueness results. Still more recently, in
\cite{A-R-S1} Alias, Romero and S\'anchez proved new uniqueness
results in spatially closed GRW spacetimes
(which includes the spatially closed Robertson-Walker spacetimes),
under the temporal convergence condition. Finally, in \cite{RRS},
Romero, Rubio and Salamanca provided uniqueness results, in the
maximal case, for spatially parabolic GRW
spacetimes, which are open models whose fiber is a parabolic
Riemannian manifold.

\section{GRW spacetimes filled with perfect fluid} \label{s4}

Currently, the interest of General Relativity in arbitrary dimension
is notable for several reasons, as the creation of unified theories
and also methodological considerations associated with the
possibility of understanding general features for the simpler
(2+l)-dimensional models (see \cite{Sha} and references therein).

Astronomical evidences indicate that the universe can be modeled (in
smoothed, averaged form) as a spacetime containing a perfect fluid
whose \emph{molecules} are the galaxies. Classically, the dominant
contribution to the energy density of the galactic fluid is the mass
of the galaxies, with a smaller pressure due mostly to radiations.
Nevertheless, over of the 90's years, evidences for the most striking
result in modern cosmology have been steadily growing, namely the
existence of a cosmological constant which is driving the current
acceleration of the universe as first observed in \cite{SP},
\cite{R}. Different models for dark energy cosmology and their
equivalences can be seen in \cite{B}. Note that a positive vacuum
energy density resulting from a cosmological constant implies a
negative pressure and viceversa.

Thus, is natural that several exact solutions to the Einstein field
equation
\begin{equation}\label{EE}
\overline{{\rm Ric}}-\frac{1}{2}\overline{S}\overline{g}=8\pi T
\end{equation}
had been obtained by considering a continuous distribution of
matter as a perfect fluid.

Recall that a \emph{perfect fluid} (see, for example, \cite[Def.
12.4]{O'N}) on a spacetime $\overline{M}$ is a triple $(U,\rho,
\mathfrak{p})$ where
\begin{enumerate}
 \item $U$ is a timelike future-pointing unit vector field on $\overline{M}$ called the \emph{flow
 vector field}.
 \item $\rho, \mathfrak{p} \in C^\infty (\overline{N})$ are, respectively, the \emph{energy density} and the
 \emph{pressure} functions.
 \item The \emph{stress-energy momentum tensor} is
 $$
 T= (\rho + \mathfrak{p} ) \, U^* \otimes U^* + \mathfrak{p} \, \overline{g} ,
 $$ where $\overline{g}$ is the metric of the spacetime $\overline{M}$.
\end{enumerate}

For an instantaneous observer $v$, the quantity $T(v,v)$ is
interpreted as the energy density, i.e. the mass-energy per unit of
volume, as measured by this observer. For normal matter, this
quantity must  be non-negative, i.e. the tensor $T$ must obey
the \emph{weak energy condition}. It is easy to see that an exact solution to (\ref{EE})
for a stress-energy tensor which obeys the weak energy condition
must satisfy the null energy condition, that is, $\overline{{\rm
Ric}}(z,z)\geq 0$ for all null vector $z$. Nevertheless, perfect fluids  can also
be used to model another scenarios of universes at the dark energy
dominated stage (see \cite{CST}).

Let us consider  a GRW spacetime $(\overline{M}, \overline{g})$ filled
with perfect fluid. As is usual in this context, we take the flow
vector field $U=\partial_t$ and therefore
$T=(\rho+\mathfrak{p})dt\otimes dt+\mathfrak{p}\overline{g}$.

Now, taking into account (\ref{Ricb}), (\ref{Sb}) and the equality
$T(\partial_t,\partial_t)=\rho$, we have
\begin{equation}\label{density}
8\pi\rho=\frac{1}{2}\left(\frac{S^ F}{f^2}+n(n-1)\frac{{f'}^
2}{f^2}\right).
\end{equation}

Analogously, if we consider vector fields $X,Y \bot\, \partial_t$,
we obtain
\[
8\pi\mathfrak{p}\overline{g}(X,Y)={\rm Ric}^
F(X,Y)+\left(\frac{f''}{f}+(n-1)\frac{f'^2}{f^2}\right)
\overline{g}(X,Y).
\]

As the pressure is a function which depends only of the points of $\overline{M}$,
the Ricci curvature of the fiber must be constant, i.e. the
Riemannian manifold  $(F,g_{_F})$ must be \emph{Einstein}. Consequently, we can
write
\begin{equation}\label{pressure}
8\pi\mathfrak{p}=\frac{1}{f^ 2}\left(Ric^F-(n-1)f^ 2(\log
f)''\right)-\frac{1}{2}\left(\frac{S^ F}{f^
2}+n(n-1)\frac{f'^2}{f^2}\right).
\end{equation}

If we assume that the spacetime obeys the NCC, then from
(\ref{density}) we have $\mathfrak{p}\geq-\rho$. In particular, when $\rho>0$ we
have
\[
\frac{\mathfrak{p}}{\rho}\geq -1.
\]

Observe that from the equations (\ref{density}) and (\ref{pressure}), the de Sitter spacetime appears
as a barotropic perfect fluid solution of the Einstein field equations with constant energy density
and pressure functions, and state equation
\[
w=\mathfrak{p}/\rho=-1.
\]

If for each $p\in F$ we parametrize $I\times \{ p \}$ by
$\gamma_{_p}(t)=(t,p)$, since $\partial_t$ is the velocity of
each \emph{galaxy} $\gamma_{_p}$, they are its integral curves. In particular,
the function $t$ is the common proper time of all galaxies.
By taking $t$ as a constant, we get the hypersurface
\[
M(t)=t\times F=\{(t,p): p\in F\}.
\]

The distance between two \emph{galaxies} $\gamma_{_p}$ and $\gamma_{_q}$ in $M(t)$ is $f(t)d(p,q)$,
where $d$ is the Riemannian distance in the fiber $F$. In particular,
when $f$ has positive derivative the spaces $M(t)$ are expanding.
Moreover, if $f''>0$ the GRW spacetimes models  universes in
accelerated expansion.

Some GRW spacetimes
satisfying NCC can be suitable modified  models of gravity.
For instance, the de Sitter spacetime  $f'(t)=\sinh t\geq 0$ (and $f'(t)=0$ only for
$t=0$) and so the spaces $M(t)$ are expanding. Moreover, $f''(t)=\cosh
t>0$ and so this expansion is accelerated. That is, the de Sitter spacetime constitutes
an accelerated expanding spacetime.

Observe that for GRW spacetimes satisfying the NCC, the
inequality $(\log f)''>0$ implies that the Ricci curvature of the
fiber is non-negative, including the Ricci flat case. When
$f'>0$, the condition $(\log f)''\geq 0$ assures accelerated
expanding models.

On the other hand, when the GRW spacetime obeys NCC, from (\ref{NCCc}) we have that the
energy density satisfies
\[
\rho\geq\frac{1}{8\pi}\frac{f''}{f}.
\]

Thus, we have
\begin{teor}
Every GRW satisfying the NCC with $f''\geq 0$ and  filled with perfect fluid
obeys the weak energy condition.
\end{teor}

In fact, given $X$ a timelike vector field we have
\[
T(X,X)=(\rho+\mathfrak{p})\overline{g}(X,\partial_t)^2+\mathfrak{p}\overline{g}(X^F,X^F),
\]
being $\overline{g}(X,\partial_t)^2\geq
\overline{g}(X^F,X^F)$, $\rho\geq 0$ and $\rho\geq
\left|\mathfrak{p}\right|.$

\subsection{Spatially closed GRW spacetimes}

Let $\overline{M}$ be a spatially closed GRW spacetime  and
$M$ a compact spacelike hypersurface in $\overline{M}$. From
Lemma \ref{l1}, we can give the following estimation of the total energy on $M$

\begin{teor}\label{tf1} Let $\psi: M \rightarrow \overline{M}$ be a compact spacelike hypersurface
in a GRW spacetime $\overline{M}= I \times_f F$ filled with perfect fluid and which is an exact solution of the Einstein field equation.
Then the total energy on $M$ satisfies that
\[
E_{_M}=\int_M\rho\,dV\leq\frac{1}{16\pi}\int_M \left(S+n(n-1)H^ 2\right)\,dV.
\]
\end{teor}

If $M$ has constant mean curvature, from (\ref{pfl}) we have
\[
\left(\frac{S^ F\circ \pi_{_F}}{f^2(\tau)}+n(n-1)\frac{{f'}^2(\tau)}{f^2(\tau)}\right)=S+n(n-1)H^2.
\]
Hence, we can calculate the total energy on the spacelike hypersurface $M$ as
\begin{eqnarray}
E_{_M} & = & \int_M\rho\,dV=\frac{1}{16\pi}\int_M \left(S+n(n-1)H^
2\right)\,dV \nonumber \\
& = & \frac{1}{16\pi}\int_M S\, dV+\frac{n(n-1)}{16\pi}H^
2{\rm Vol}(M).  \label{TE}
\end{eqnarray}
In particular, if $M$ is maximal, then
\[
E_{_M}=\frac{1}{16\pi}\int_M S\, dV.
\]

\begin{coro} Let $\psi: M \rightarrow \overline{M}$ be a compact CMC spacelike hypersurface
in a GRW spacetime $\overline{M}= I \times_f F$ obeying the
NCC, filled with perfect fluid and which is an exact solution of the Einstein field equation. Then $M$ is maximal if and only if the total energy on $M$ coincides with the
integral on the hypersurface of its scalar curvature.
\end{coro}

On the other hand, observe that if the fiber of $\overline{M}$ has
non-negative scalar curvature, then the total energy on every
compact maximal hypersurface is non-negative.

Taking into account (\ref{TE}) and \cite[Th. 3]{ARR2}, we can give the
following estimations of the total energy on a CMC compact spacelike
hypersurface in terms of the volume of the fiber of the spacetime

\begin{teor}\label{tf2} Let $\psi: M \rightarrow \overline{M}$ be a compact CMC spacelike hypersurface
in a GRW spacetime $\overline{M}= I \times_f F$ obeying the NCC, filled
with perfect fluid and which is an exact solution of the Einstein
field equation. Then
\[
\frac{1}{16\pi}\int_M S\, dV+\frac{n(n-1)f(\tau_0)^n}{16\pi\cosh\varphi^0}H^ 2\, {\rm Vol}(F)\leq E_{_M}
\leq  \frac{1}{16\pi}\int_M S\, dV+ \frac{n(n-1)f(\tau^0)^n}{16\pi\cosh\varphi_0}H^ 2\, {\rm Vol}(F),
\]
were, $f(\tau^0)$ (resp. $f(\tau_0)$) denotes the maximum
(resp. the minimum) of the warping function on the
hypersurface and $\cosh\varphi^0$ (resp. $\cosh\varphi_0$) denotes
the maximum (resp. the minimum) on $M$ of the cosine of $\cosh\varphi$.
\end{teor}

Its is remarkable that in the case (2+1)-dimensional the formulae (\ref{TE}) has an special significance,
since $S=2K$, were $K$ denotes the Gaussian curvature of $M$.
Hence, using the Gauss-Bonet theorem, we get

\begin{teor}\label{GB} Let $\psi: M \rightarrow \overline{M}$ be a compact CMC spacelike hypersurface
in a 3-dimensional GRW spacetime $\overline{M}= I \times_f F$ obeying the
NCC, filled with perfect fluid and which is an exact solution of the Einstein field equation. Then
\[
E_{_M}=\frac{1}{8}\mathcal{X}(M)+\frac{n(n-1)}{16\pi}H^ 2{\rm Vol}(M),
\]
where $\mathcal{X}(M)$ is the Euler characteristic of $M$.
\end{teor}

Thus, the total
energy is explained in terms of topological and extrinsic
quantities.

\begin{rem}
Observe that taking into account the relation $\rho+\mathfrak{p}\geq 0$ and the
previous considerations, some estimations for the total pressure on a
compact spacelike hypersurface can also be given.
\end{rem}

\section*{Acknowledgements}
The first author is partially supported by the Spanish MICINN Grant
with FEDER funds MTM2010-19821. The second  author is
partially supported by the Spanish MICINN Grant with FEDER funds
MTM2010-18099.

\end{document}